\documentclass[11pt]{article}
\usepackage[dvipsnames,svgnames,x11names,hyperref]{xcolor}
\usepackage{fullpage}
\usepackage{amsmath}
\usepackage{amssymb}
\usepackage{amsthm}
\usepackage{thmtools}
\usepackage[colorlinks]{hyperref}

\usepackage{mathpazo}
\usepackage{todonotes}
\presetkeys{todonotes}{size=\footnotesize,color=green!40}{}
\numberwithin{equation}{section}
\usepackage{mdframed}
\usepackage{tikz}

\usepackage{scrextend} 

\hypersetup{
  linkcolor=[rgb]{0.3,0.3,0.6},
  citecolor=[rgb]{0.2, 0.6, 0.2},
  urlcolor=[rgb]{0.8,0.4,0.2}
}

\usetikzlibrary{arrows}

\declaretheoremstyle[bodyfont=\it,qed=$\qed$]{noproofstyle} 

\declaretheorem[numberwithin=section]{theorem}
\declaretheorem[numberlike=theorem]{lemma}
\declaretheorem[numberlike=theorem]{corollary}

\declaretheorem[numberlike=theorem]{claim}

\declaretheorem[numberlike=theorem]{subclaim}

\declaretheorem[unnumbered,name=Theorem,style=noproofstyle]{theorem*}
\declaretheorem[unnumbered,name=Lemma,style=noproofstyle]{lemma*}
\declaretheorem[unnumbered,name=Corollary,style=noproofstyle]{corollary*}
\declaretheorem[unnumbered,name=Proposition,style=noproofstyle]{proposition*}

\declaretheorem[unnumbered,name=Claim]{claim*}
\declaretheorem[unnumbered,name=Conjecture]{conjecture*}
\declaretheorem[unnumbered,name=Question]{question*}

\declaretheoremstyle[bodyfont=\it]{defstyle} 
\declaretheorem[numberlike=theorem,style=defstyle]{definition}

\declaretheorem[unnumbered,name=Remark,style=defstyle]{remark*}

\newenvironment{myproof}[1]%
{\vspace{1ex}\noindent{\em Proof of #1.}\hspace{0.5em}\def\myproof@name{#1}}%
{\hfill{\tiny \qed\ (\myproof@name)}\vspace{1ex}}

\newenvironment{proof-sketch}{\medskip\noindent{\em Sketch of Proof.}\hspace*{1em}}{\qed\bigskip}
\newenvironment{proof-attempt}{\medskip\noindent{\em Proof attempt.}\hspace*{1em}}{\bigskip}

\newcommand{\inparen }[1]{\left(#1\right)}             
\newcommand{\inbrace }[1]{\left\{#1\right\}}           
\newcommand{\inangle }[1]{\left\langle#1\right\rangle} 
\newcommand{\pfrac}[2]{\inparen{\frac{#1}{#2}}}

\newcommand{\abs}[1]{\left|#1\right|}                  
\newcommand{\norm}[1]{\left\Vert#1\right\Vert}         

\newcommand{\union}{\cup}
\newcommand{\Union}{\bigcup}

\newcommand{\eqdef}{\stackrel{\text{def}}{=}}
\newcommand{\setdef}[2]{\inbrace{{#1}\ : \ {#2}}}      

\newcommand{\poly}{\mathrm{poly}}

\newcommand{\veca}{\mathbf{a}}

\newcommand{\vecu}{\mathbf{u}}
\newcommand{\vecv}{\mathbf{v}}

\newcommand{\spaced}[1]{\quad #1 \quad}

\newcommand{\NP}{\mathsf{NP}}

\renewcommand{\epsilon}{\varepsilon}
\newcommand{\val}{\mathrm{val}}

\newcommand{\sym}{\text{sym}}

\begin{document}
\title{On Fortification of Projection Games}
\author{
  Amey Bhangale\thanks{Rutgers University. Research supported by the NSF grant CCF-1253886. {\tt amey.bhangale@rutgers.edu}}%
  \and%
  Ramprasad Saptharishi\thanks{Tel Aviv University. The research leading to these results has received funding from the European Community's Seventh Framework Programme (FP7/2007-2013) under grant agreement number 257575. {\tt ramprasad@cmi.ac.in}}%
  \and%
  Girish Varma\thanks{Tata Institute of Fundamental Research, Mumbai. Supported by the Google India Fellowship in Algorithms. {\tt girishrv@tifr.res.in}}%
  \and%
  Rakesh Venkat\thanks{Tata Institute of Fundamental Research, Mumbai. {\tt rakesh@tifr.res.in}}%
}
\maketitle

\begin{abstract}
  A recent result of Moshkovitz~\cite{Moshkovitz14} presented an
  ingenious method to provide a completely elementary proof of the
  \emph{Parallel Repetition Theorem} for certain projection games via
  a construction called \emph{fortification}. However, the
  construction used in \cite{Moshkovitz14} to fortify arbitrary label
  cover instances using an arbitrary extractor is insufficient to
  prove parallel repetition. In this paper, we provide a fix by using
  a stronger graph that we call \emph{fortifiers}. Fortifiers are graphs that have both
  $\ell_1$ and $\ell_2$ guarantees on induced distributions from large
  subsets. 

  We then show that an expander with sufficient spectral gap, or a
  bi-regular extractor with stronger parameters (the latter is also the
  construction used in an independent update \cite{Moshkovitz15} of
  \cite{Moshkovitz14} with an alternate argument), is a good
  fortifier. We also show that using a fortifier (in particular $\ell_2$ guarantees)  is necessary for
  obtaining the robustness required for fortification. 

\end{abstract}


\section{Introduction}

\subsection*{Label-cover and general two-prover games}

A label cover instance is specified by a bipartite graph $G =
((X,Y),E)$, a pair of alphabets $\Sigma_X$ and $\Sigma_Y$ and a set of
constraints $\psi_e: \Sigma_X \rightarrow \Sigma_Y$ on each edge $e\in
E$. The goal is to label the vertices of $X$ and $Y$ using
labels from $\Sigma_X$ and $\Sigma_Y$ so as to satisfy as many
constraints are possible.

This problem is often viewed as a two-prover game. The verifier picks
an edge $(x,y)$ at random and sends $x$ to the first prover and $y$ to
the second prover. They are to return a label of the vertex that they
received, and the verifier accepts if the labels they returned are
consistent with the constraint $\psi_{(x,y)}$. The value of this game
$G$, denoted by $\val(G)$, is given by the acceptance probability of
the verifier maximized over all possible strategies of the
provers. These are also called \emph{projection games} as the
constraints are functions from $\Sigma_X$ to $\Sigma_Y$. They are
called \emph{general games} if the constraint on each edge is an
arbitrary relation $\psi_{(x,y)} \subseteq \Sigma_X \times \Sigma_Y$. 

These two notions are equivalent in the sense that $\val(G)$ is
exactly equal to the maximum fraction of constraints that can be
satisfied by any labelling.\\

This problem is central to the PCP Theorem \cite{AS98,ALMSS98} and
almost all inapproximability results that stem from it. The (Strong)
PCP Theorem can be rephrased as stating that for every $\epsilon > 0$,
it is $\NP$-hard to distinguish whether a given label cover instance
has $\val(G) = 1$ or $\val(G) < \epsilon$. An important step is a way
to transform instances with $\val(G)<1-\epsilon$ to instances $G'$
with $\val(G') < \epsilon$. This is usually achieved via the
\emph{Parallel Repetition Theorem}.

\subsection*{Parallel Repetition}

The $k$-fold repetition of a game $G$, denoted by $G^k$, is the following natural
definition --- the verifier picks $k$ edges $(x_1,y_1),\cdots,
(x_k,y_k)$ from $E$ uniformly and independently, sends $(x_1,\dots, x_k)$ and $(y_1,\dots,
y_k)$ to the provers respectively, and accepts if the labels returned
by them are consistent on each of the $k$ edges.

If $\val(G)=1$ to start with then $\val(G^k)$ still remains $1$. How does $\val(G^k)$ decay with $k$ if $\val(G) < 1$? Turns out even this simple
operation of repeating a game in parallel has a counter-intuitive
effect on the value of the game. It is easy to see that $\val(G^k) \geq \val(G)^k$
as provers can use a same strategy as in $G$ to answer each query $(x_i,y_i)$.
The first surprise is $\val(G^k)$ is
\emph{not} $\val(G)^k$, but sometimes can be \emph{much larger} than
$\val(G)^k$. Fortnow~\cite{For89} presented a game $G$ for which $\val(G^2) > \val(G)^2$, 
Feige~\cite{Fei91} improved this by giving an example of game $G$ with $\val(G) < 1$ but $\val(G^2) = \val(G)$. 
Indeed, there are known examples \cite{Raz11} of projection games
where $\val(G) = (1 - \epsilon)$ but $\val(G^k) \geq \inparen{1 - \epsilon\sqrt{k}}$ 
for a large range of $k$.

The first non trivial upper bound on $\val(G^k)$ was proven by Verbitsky~\cite{Ver96} who showed that
if $\val(G)<1$ then the value $\val(G^k)$  must go to zero as $k$ goes to infinity. 
It is indeed true that $\val(G^k)$ decays exponentially
with $k$ (if $\val(G) < 1$). This breakthrough was first proved by Raz
\cite{Raz98}, and has subsequently seen various simplifications and
improvements in parameters \cite{Hol09,Rao11,DS14,BG14}.
The following statements are due to Holenstein \cite{Hol09}, Dinur and
Steurer~\cite{DS14} respectively.

\begin{theorem}[Parallel repetition theorem for general games]
  Suppose $G$ is a two-prover game such that $\val(G) \leq
  1 - \epsilon$ and let $\abs{\Sigma_X}\abs{\Sigma_Y} \leq s$. Then, for any $k \geq 0$,
  \[
  \val(G^k) \spaced{\leq} \inparen{1 - \epsilon^3/2}^{\Omega(k/\log s)}.
  \]
\end{theorem}

\begin{theorem}[Parallel repetition theorem for projection games]
  Suppose $G$ is a projection game such that $\val(G) \leq
  \rho$. Then, for any $k \geq 0$,
  \[
  \val(G^k) \spaced{\leq} \pfrac{2\sqrt{\rho}}{1 + \rho}^{k/2}.
  \]
\end{theorem}

Although a lot of these results are substantial simplifications of earlier
proofs, they continue to be involved and delicate. Arguably, one might
still hesitate to call them \emph{elementary} proofs.\\

Recently, Moshkovitz~\cite{Moshkovitz14} came up with an ingenious
method to prove a parallel repetition theorem for certain projection
games by slightly modifying the underlying game via a process that the author
called \emph{fortification}. The method of fortification suggested in
\cite{Moshkovitz14} was a rather mild change to the underlying game
and proving parallel repetition for such \emph{fortified projection
  games} was sufficient for most applications. The advantage of
fortification was that parallel repetition theorem for fortified games
had a simple, elementary and elegant proof as seen in
\cite{Moshkovitz14}.

\subsection{Fortified games}\label{sec:intro-fortified-games}

Fortified games will be described more formally in
\autoref{sec:prelims}, but we give a very rough overview
here. Moshkovitz showed that there is an easy way to bound the value
of repeated game if we knew that the game was \emph{robust on large
  rectangles}. We shall first need the notion of \emph{symmetrized projection games}.\\

{\bf Symmetrized Projection games. } Given a projection game $G$ on
$((X,Y),E)$, the symmetrized game $G_\sym$ is a game on $((X,X),E')$
such that for every $y \in Y$ with $(x,y), (x',y)\in E$, there is an
edge $(x,x') \in E'$ with the constraint $\pi_{(x,y)}(\sigma_x) =
\pi_{(x',y)}(\sigma_{x'})$.\\

For projection games, it would be more convenient to work with the
above symmetrized version for reasons that shall be explained
shortly. It is not hard to see that $\val(G)$ and $\val(G_\sym)$ are
within a quadratic factor of each other. Thus for projection games, we
shall work with the game $G_\sym$ instead of the original game $G$. 

\medskip

\begin{definition}[$(\delta,\epsilon)$-robust games] Let $G$ be a two-prover game
  on $((X,Y),E)$. For any pair of sets $S\subseteq X, T \subseteq Y$, let 
  $G_{S\times T}$ be the game where the verifier chooses his
  random query $(x,y)\in E$ conditioned on the event
  that $x \in S$ and $y \in T$.

  $G$ is said to be \emph{$(\delta,\epsilon)$-robust} if for
  every $S,T \subseteq X$ with $|S|\geq \delta |X|$ and $|T| \geq \delta |Y|$ we have that
  \[
  \val(G_{S\times T}) \quad\leq \quad \val(G) + \epsilon.
  \]
\end{definition}

\begin{theorem}[Parallel repetition for robust projection games \cite{Moshkovitz14}]\label{thm:robustness-to-pr} Let
  $G$ be a projection game on a bi-regular bipartite graph $((X,Y),E)$
   with alphabets $\Sigma_X$ and
  $\Sigma_Y$. For any positive integer $k$, if $\epsilon,\delta > 0$
  are parameters such that $2\delta |\Sigma_Y|^{k-1} \leq \epsilon$ and
  $G_\sym$ is $(\delta,\epsilon)$-robust, then\footnote{The following is the corrected statement from \cite{Moshkovitz15}.}
  \[
  \val(G_\sym^k) \spaced{\leq} \inparen{\val(G_\sym)+\epsilon}^k + k\epsilon.
  \]
\end{theorem}

Not all projection games are robust on large rectangles, but Moshkovitz
suggested a neat way of slightly modifying a projection game and
making it robust. This process was called \emph{fortification}.

On a high level, for any two-prover game, the verifier chooses to
verify a constraint corresponding to an edge $(x,y)$ but is instead
going to sample several other dummy vertices and give the provers two
sets of $D$ vertices $\inbrace{x_1,\dots, x_D}$ and
$\inbrace{y_1,\dots, y_D}$ such that $x = x_i$ and $y = y_j$ for some
$i$ and $j$. The provers are expected to return labels of all $D$
vertices sent to them but the verifier checks consistency on just the
edge $(x,y)$. This is very similar to the ``confuse/match''
perspective of Feige and Kilian~\cite{FK94}.

To derandomize this construction, Moshkovitz~\cite{Moshkovitz14} uses
a pseudo-random bipartite graph where given a vertex
$w$, the provers are expected to return labels of all its
neighbours (\autoref{defn:concatenation}). The most natural
candidate of such a pseudo-random graph is  an $(\delta,
\epsilon)$-extractor, as we really want to ensure that conditioned on
``large enough events'' $S$ and $T$, the underlying distribution on the
constraints does not change much. This makes a lot of intuitive sense,
since on choosing a random element of $S$ and then a random neighbour,
the extractor property guarantee that the induced distribution on
vertices in $X$ is $\epsilon$-close to uniform. Thus, it is natural to
expect that conditioning on the events $S$ and $T$ should not change
the underlying distribution on the constraints by more than
$O(\epsilon)$. This was the rough argument in \cite{Moshkovitz14},
which unfortunately turns out to be false. We elaborate on this in
\autoref{sec:extractor-insufficient} and \autoref{sec:explicit-extractor-counterexample}. 

A recent updated version \cite{Moshkovitz15} of \cite{Moshkovitz14} provides
an different argument for the fortification lemma using a stronger
extractor. We discuss this at the end of \autoref{sec:contributions}.

\subsection{Our contributions}\label{sec:contributions}

We present a fix to the approach of \cite{Moshkovitz14}, by describing
a way to transform any given game instance $G$ into a robust instance
$G^*$ with the same value following the framework of
\cite{Moshkovitz14} but using a different graph for concatenation, and
a different analysis.

We first describe a concrete counter-example to the original argument
of \cite{Moshkovitz14} in \autoref{sec:extractor-insufficient}, that
shows  concatenating (\autoref{defn:concatenation})
with an arbitrary $(\delta,\epsilon)$-extractor is insufficient. In
fact, as we show in \autoref{sec:random-graph}, concatenating with \emph{any}
left-regular graph with left-degree by $o(1/\epsilon\delta)$ fails to
make arbitrary instances $(\delta, \epsilon)$-robust. We instead use
bipartite graphs called \emph{fortifiers}, defined below.

\begin{definition}[Fortifiers]\label{defn:fortifiers}
  A bipartite graph $H = ((W,X),E_H)$ is an
  \emph{$(\delta, \epsilon_1, \epsilon_2)$-fortifier} if for any set
  $S \subseteq W$ such that $|S| \geq \delta|W|$, if $\pi$ is the
  probability distribution on $X$ induced by picking a uniformly
  random element $w$ from $S$, and a uniformly random neighbor $x$ of $w$, then
  \begin{eqnarray*}
    \abs{\pi - \vecu}_1 & \leq & \epsilon_1,\\
    \norm{\pi - \vecu}^2 & \leq & \frac{\epsilon_2}{|X|}.
  \end{eqnarray*}  
\end{definition}

Notice that a fortifier is an extractor, with the additional
condition that the $\ell_2$-distance of $\pi$ from the
uniform distribution is small. This is what enables us to show that
concatenation with a fortifier produces a robust instance. 

\begin{theorem}[Fortifiers imply robustness]\label{thm:fortification-projection}
Suppose $G$ is a general two-prover game on a bi-regular graph $((X,Y),E)$. Then, for any
$\epsilon, \delta >0$, if $H_1 = ((W,X),E_1)$ and $H_2 = ((Z,Y),E_2)$ are $(\delta, \epsilon,
  \epsilon)$-fortifiers, then the concatenated game
  $G^* = H_1 \circ G \circ H_2$ is $(\delta,O(\epsilon))$-robust.
\end{theorem}

In particular, bipartite spectral expanders are good fortifiers, as
\autoref{lem:expanders-are-fortifiers} shows. This gives us our main
result which follows from \autoref{lem:expanders-are-fortifiers} and \autoref{thm:fortification-projection}: 

\begin{corollary}\label{thm:main-thm}
  Let $G$ be a general two-prover game on a bi-regular graph
  $((X,Y),E)$. For any $\epsilon,\delta > 0$, if $H_1 = ((W,X),E_1)$
  and $H_2 = ((Z,Y),E_2)$ are two $\lambda$-expanders (\autoref{defn:expanders}) 
  with $\lambda < \epsilon \sqrt{\delta}$
  then concatenated game $G^*= H_1 \circ G \circ H_2$ is
  $(\delta,4\epsilon)$-robust.
\end{corollary}

As one would expect, the condition on the fortifier can be relaxed if
the underlying graph of the original label cover instance is a
spectral-expander. We prove the following theorem.
 \autoref{thm:fortification-projection}
follows from this theorem by setting $\lambda_0=1$.
\begin{theorem}\label{thm:fortification-on-expanders}
  Let $G$ be a two-prover game on bi-regular graph $((X,Y),E)$ where
  $G$ is an $\lambda_0$-expander. Then for any $\epsilon,\delta > 0$,
  if $H_1 = ((W,X),E_1)$ and $H_2 = ((Z,Y),E_2)$ are $(\delta, \epsilon,
  (\epsilon/\lambda_0))$-fortifiers, then the concatenated game $G^* =
  H_1 \circ G \circ H_2$ is $(\delta,O(\epsilon))$-robust.
\end{theorem}

One could ask if the definition of a fortifier is too strong, or if a
weaker object would suffice. We argue
in \autoref{sec:fortifiers-are-necessary}
that if we proceed through concatenation, fortifiers are indeed
necessary to make a game robust. 

Bipartite Ramanujan graphs of degree $\Theta(1/\epsilon^2 \delta)$
have $\lambda < \epsilon \sqrt{\delta}$ and are therefore good
fortifiers. In \autoref{sec:random-graph}, we show that this is almost
optimal by proving a lower bound of $\Omega(1/\epsilon\delta)$ on the
left-degree of any graph that can achieve $(\delta,
\epsilon)$-robustness. This shows that our construction of using
expanders to achieve robustness is almost optimal, in terms of the
degree of the fortifier graph. Note that the degree of the fortifier is important
as the alphabet size of the concatenated game is the alphabet size of
the original game raised to the degree. There are known explicit constructions of 
bi-regular $(\delta, \epsilon)$-extractors with left-degree
$\poly(1/\epsilon) \poly\log(1/\delta)$. But the lower bound in \autoref{sec:fortifiers-are-necessary}
shows that $(\delta,\epsilon)$-extractors are not fortifiers if $\delta
\ll \epsilon$, which is usually the relevant setting (see \autoref{thm:robustness-to-pr} and \autoref{thm:parrep-for-general}). 

Though all the above results are stated for bi-regular games, any
two-prover game can be easily converted to one on a bi-regular graph
or roughly the same value via standard tricks. We outline such a
construction (similar to the construction in \cite{DH13} for
projection games) in \autoref{sec:deg-reduction}.\\

Independently, the author of \cite{Moshkovitz14} came up with a
different argument to obtain robustness of projection games by using a
$(\delta, \epsilon\delta)$-extractor. This is described in an updated
version \cite{Moshkovitz15} present on the author's homepage.

It is also seen from \autoref{thm:fortification-on-expanders} that
bi-regular $(\delta, \epsilon\delta)$-extractors are indeed $(\delta,
\epsilon,\epsilon)$-fortifiers as well. Using an expander instead is arguably
simpler, and is almost optimal.

\begin{remark*}
  Although this fix provides a proof of a Parallel Repetition Theorem
  for projection games following the framework of \cite{Moshkovitz14},
  the degree of the fortifier is too large to get the required PCP for proving optimal hardness
  of the \textsc{Set-Cover} problem that Dinur and Steurer~\cite{DS14}
  obtained. See \cite{Moshkovitz15} for a discussion on this. 
\end{remark*}

\subsubsection*{Remark about parallel repetition for general games}

A fairly straightforward generalization \autoref{thm:robustness-to-pr}
to robust general games on bi-regular graphs is the following.

\begin{lemma}[Parallel repetition for general robust games]\label{thm:parrep-for-general}
  Let $G$ be a general two-prover game on a bi-regular graph $((X,Y),E)$ with alphabets
  $\Sigma_X$ and $\Sigma_Y$. For any positive integer $k$, if
  $\epsilon,\delta > 0$ are parameters such that $2\delta
  |\Sigma_X \times \Sigma_Y|^{k-1} \leq \epsilon$ and $G$ is
  $(\delta,\epsilon)$-robust, then
  \[
  \val(G^k) \spaced{\leq} \inparen{\val(G)+\epsilon}^k + k\epsilon.
  \]  
\end{lemma}

But it is to be noted that the fortification procedure via
concatenating a fortifier makes $|\Sigma_X| = \exp(1/\delta)$
and in such scenarios $\delta |\Sigma_X| \gg 1$ making it infeasible
to ensure $2\delta|\Sigma_X \times \Sigma_Y|^{k-1} \leq
\epsilon$. Hence, though \autoref{thm:parrep-for-general} may be
useful in cases where we know that the game $G$ is robust via other
means, the technique of fortification via concatenation increases the
alphabet size too much for \autoref{thm:parrep-for-general} to be applicable.

\medskip

For the case of projection games, this is not an issue as
\autoref{thm:robustness-to-pr} only requires $2\delta |\Sigma_Y|^{k-1}
< \epsilon$ and concatenating $G_\sym$ by a fortifier only increases
$\abs{\Sigma_X}$ and keeps $\Sigma_Y$ unchanged. Thus, one can indeed
choose $\epsilon$ and $\delta$ small enough to give a parallel
repetition theorem for a robust version of an arbitrary projection
game. 

\section{Preliminaries}\label{sec:prelims}
\subsection*{Notation}

\begin{itemize}
  \setlength\itemsep{0em}
\item For any vector $\veca$, let $\abs{\veca}_1 := \sum_i \abs{\veca_i}$, 
and $\norm{\veca} := \sqrt{ \sum_i \veca_i^2}$ be the $\ell_1$ and $\ell_2$-norms respectively.
\item We shall use $\vecu_S$ to refer to the uniform distribution on a
  set $S$. Normally, the set $S$ would be clear from context and in
  such case we shall drop the subscript $S$.
\item For any vector $\veca$, we shall use $\veca^{\parallel}$ to
  refer to the component along the direction of $\vecu$, and
  $\veca^\perp$ to refer to the component orthogonal to $\vecu$.
\item We shall assume that the underlying graph for the games is
  bi-regular.  This is more or less without loss of generality via
  standard sampling tricks (see \autoref{sec:deg-reduction}).
\end{itemize}

We define the \emph{concatenation} operation of a two-prover games  with a bipartite graph  that was
alluded to in \autoref{sec:intro-fortified-games}.

\begin{definition}[Concatenation]\label{defn:concatenation}

Given a two-prover game on a graph $G=((X,Y),E)$ with a set of constraints $\psi$, a pair of alphabets $\Sigma_X$ and $\Sigma_Y$, bipartite
graphs $H_1=((W,X),E_1)$ with left degree $D_1$, and $H_2 = ((Z,Y),E_2)$ with left-degree $D_2$, the \emph{concatenated game} is a  game on the (multi) graph $H_1 \circ G \circ H_2 = ((W,Z),E_{H_1\circ G \circ H_2})$  with $\Sigma_W =
\Sigma_X^{D_1}$ and $\Sigma_Z = \Sigma_Y^{D_2}$. Label of a vertex $w\in W$ ($z\in Z$) can be thought of labels to its neighbors in $H_1$($H_2$) in a fixed order. For any edge $(w,z) \in E_{H_1\circ G \circ H_2}$, there exists $ (x,y) \in E$ such that $(w,x) \in E_1$, and $(z,y) \in E_2$. The constraint for this
edge first obtains the label of $x$ from $w$, and similarly obtains the label for $y$ from the label of $z$, and checks the constraint $\psi_{(x,y)}$ according to the game $G$.
\end{definition} 

\begin{figure}
\begin{center}
\begin{tikzpicture}
\draw (0,0) ellipse (0.5cm and 2cm);
\draw (2,0) ellipse (0.5cm and 2cm);
\draw (4,0) ellipse (0.5cm and 1.5cm);
\draw (6,0) ellipse (0.5cm and 1.5cm);

\node at (0,-2.5) {$W$};
\node at (2,-2.5) {$X$};
\node at (4,-2) {$Y$};
\node at (6,-2) {$Z$};
\tikzstyle{mycirc}=[circle, draw,
                        inner sep=0pt, minimum width=4pt]
\node[mycirc] (w1) at (0,1.5) {};
\node[mycirc] (w2) at (0,0.5) {};
\node[mycirc] (w3) at (0,-0.5) {};
\node[mycirc] (w4) at (0,-1.5) {};

\node[mycirc] (x1) at (2,1.5) {}
edge[<-,thin,draw=black!30] (w1)
edge[<-,thin,draw=black!30] (w3)
;
\node[mycirc] (x2) at (2,0.5) {}
edge[<-,thin,draw=black!30] (w1)
edge[<-,thin,draw=black!30] (w4)
;

\node[mycirc] (x3) at (2,-0.5) {}
edge[<-,thin,draw=black!30] (w3)
edge[<-,thin,draw=black!30] (w2)
;

\node[mycirc] (x4) at (2,-1.5) {}
edge[<-,thin,draw=black!30] (w2)
edge[<-,thin,draw=black!30] (w4)
;

\node[mycirc] (y1) at (4,1) {}
edge[<-,thin,draw=black!30] (x1)
edge[<-,thin,draw=black!30] (x2)
edge[<-,thin,draw=black!30] (x3)
edge[<-,thin,draw=black!30] (x4)
;

\node[mycirc] (y2) at (4,0) {}
edge[<-,thin,draw=black!30] (x1)
edge[<-,thin,draw=black!30] (x2)
edge[<-,thin,draw=black!30] (x3)
edge[<-,thin,draw=black!30] (x4)
;

\node[mycirc] (y3) at (4,-1) {}
edge[<-,thin,draw=black!30] (x1)
edge[<-,thin,draw=black!30] (x2)
edge[<-,thin,draw=black!30] (x3)
edge[<-,thin,draw=black!30] (x4)
;

\node[mycirc] (z1) at (6,1) {}
edge[<-,thin,draw=black!30] (y2)
edge[<-,thin,draw=black!30] (y3)
;
\node[mycirc] (z2) at (6,0) {}
edge[<-,thin,draw=black!30] (y1)
edge[<-,thin,draw=black!30] (y3)
;

\node[mycirc] (z3) at (6,-1) {}
edge[<-,thin,draw=black!30] (y1)
edge[<-,thin,draw=black!30] (y2)
;

\draw[very thick,->] (w2) to[bend left] (z3);
\path (w2) edge[thick,dashed, ->] (x3)
 (x3) edge[thick,dashed, ->] (y2) 
 (y2) edge[thick,dashed, ->] (z3);
\end{tikzpicture}
\end{center}
\caption{Concatenated Games}\label{fig:concat-games}
\end{figure}
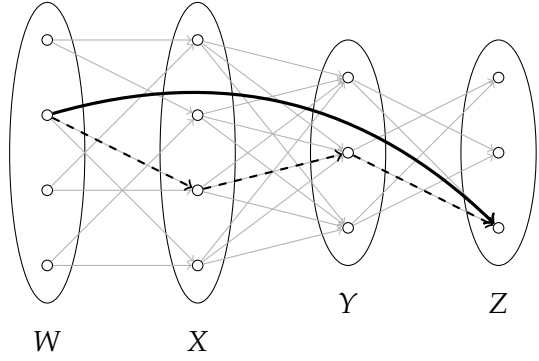

\begin{remark*}
  As mentioned earlier for projection games, as in \cite{Moshkovitz14}
  we shall work with symmetrized version $G_\sym$. In $G_\sym$ which
  is played on $((X, X),E_\sym)$,
  concatenating both sides with the same $H_1$ ensures that the resulting game $G^*$ is
  still a symmetrized projection game, and that
  the concatenation operation only changes $\Sigma_X$ and leaves
  $\Sigma_Y$ unchanged for the underlying projection game.

  We state the results in a general setting as the focus here would be
  mainly on the study of distributions of edges of sub-graphs of
  concatenated graphs.
\end{remark*}

\begin{lemma}[Concatenation preserves value]\label{lem:concat-value-same} \cite{Moshkovitz14}
Given any two-prover game on a bi-regular graph $G$, if $H_1$ and $H_2$ are bi-regular graphs, then we have:
\[
\val(H_1\circ G\circ H_2) = \val(G).
\]
\end{lemma}

\subsection*{Expanders, extractors and fortifiers}

\begin{definition}[Expanders]\label{defn:expanders}
  For any bi-regular bipartite graph $H = ((X,Y),E)$ with $|X| = |Y|$ and (left) degree $D$, we shall use
  $\lambda(H)$ to denote
\[
\lambda(H) \eqdef \max_{\vecv \perp \vecu} \frac{\norm{H\vecv}}{\norm{\vecv}}
\]
where the matrix $H$ is an $|Y| \times |X|$ matrix (rows indexed by vertices in $Y$, and columns by vertices in $X$) defined by $H(y,x) = 1/D$ if $(x,y)\in E$ and it is $0$ otherwise.
For any $\lambda > 0$, a bi-regular bipartite graph $H$ is an
\emph{$\lambda$-expander} if $\lambda(H) \leq \lambda$. 

More generally\footnote{We are not sure if this definition is standard, but is a natural generalization and precisely what we need in our proof.}, if $|X| \neq |Y|$, we define $\lambda(H)$ as follows:
\[
\lambda(H) \;\eqdef\; \max\limits_{\vecv \perp \vecu} \frac{\norm{H\vecv}}{\norm{\vecv}} \cdot \inparen{\frac{\norm{\vecu_X}}{\norm{H\vecu_X}}}.
\]
\end{definition}

\noindent 
Informally, $\lambda(H)$ measures ``how much \emph{more} does the matrix $H$ shrink $\vecv\perp \vecu_X$ compared to $\vecu_X$''?

\begin{lemma}[Explicit expanders \cite{BL06}]\label{lem:expander-construction}
For every $D >0$, there exists a fully explicit family of bipartite graphs $\{G_i\}$, such that $G_i$ is $D$-regular on both sides and $\lambda(G_i) \leq D^{-1/2} (\log D)^{3/2}$.
\end{lemma}

\begin{definition}[Extractors]
  A bipartite graph $H = ((X,Y),E)$ is an
  $(\delta,\epsilon)$-extractor if for every subset $S \subseteq X$
  such that $|S| \geq \delta |X|$, if $\pi$ is the induced probability
  distribution on $Y$ by taking a random element of $S$ and a random
  neighbour, then
  \[
  \abs{\pi - \vecu}_1 \spaced{\leq}  \epsilon.
  \]
\end{definition}

\begin{lemma}[Explicit Extractors \cite{RVW00}]\label{lem:extractor-construction}
There exists explicit $(\delta, \epsilon)$-extractors $G=(X,Y,E)$ such that $|X| = O(|Y|/\delta)$ and each vertex
of $X$ has degree $D= O(\exp(\poly(\log \log (1/\delta)))\cdot(1/\epsilon^2))$. 
\end{lemma}

Our earlier definition of a fortifier (\autoref{defn:fortifiers}) has
properties of both an expander and an extractor. Indeed, we can build
fortifiers by just taking a product an expander and an extractor.

\begin{lemma}\label{lem:building-fortifiers}
  If $H_1 = ((V,W),E_1)$ is a bi-regular
  $(\delta,\epsilon)$-extractor, and if $H_2 = ((W,X),E_2)$ is a bi-regular
  $\lambda$-expander, then the product graph $H_1 \cdot  H_2$ is an $(\delta,\epsilon, \lambda^2\epsilon/\delta)$-fortifier.
\end{lemma}
\begin{proof} 
Let $H_2$ be the normalized adjacency matrix of graph $H_2$ and let $\pi_S$ denote the probability distribution on $W$ obtained by picking an element of $S\subseteq V$ uniformly and then choosing  a random neighbour in $H_1$. Thus, $H_2\pi_S$ is the probability distribution on $X$ induced by the uniform distribution on $S$ and a random neighbour in $H_1 \cdot H_2$. We want to show for all $S$ such that $|S|\geq \delta |V|$,
$$| H_2\pi_S - \vecu|_1 \leq \epsilon \text{\; and\; } \|  H_2\pi_S - \vecu\|^2 \leq \frac{\lambda^2\epsilon/\delta}{|X|}.$$
The first inequality is obtained as $| H_2\pi_S - \vecu|_1 = | H_2(\pi_S - \vecu)|_1\leq  |\pi_S - \vecu|_1 \leq \epsilon,$
where we use the fact that $\abs{H_2 v}_1 \leq \abs{v}_1$ for any $v$ and any normalized adjacency matrix, and $|\pi_S - \vecu|_1 \leq \epsilon$ follows form the extractor property of $H_1$.

\noindent
As for the second inequality, observe that
$$\norm{\pi_S - \vecu}^2 \leq \max_{w\in W}(\pi_S(w)) \cdot |\pi_S - \vecu|_1 \leq \epsilon \cdot \max_{w\in W}(\pi_S(w)).$$
For a bi-regular extractor\footnote{The bound on the right-degree guaranteed by bi-regularity is crucial for this claim. Without this, extractors are not sufficient for fortification (\autoref{sec:extractor-insufficient}).} $H_1$ of left-degree $D$, the degree of any $w\in W$ is $\inparen{|V|\cdot D/|W|}$ and the number of edges out of $S$ is least $\delta |V| \cdot D$. Hence, $\max_w \pi_S(w) \leq 1/(\delta |W|)$, which is achieved if all neighbours of $w$ are in $S$. Therefore, 
\begin{align*}
\norm{\pi_S - \vecu}^2 & \leq  \frac{(\epsilon/\delta)}{|W|}\\
\implies \norm{H_2(\pi_S - \vecu)}^2 & \leq  \lambda^2\frac{|W|}{|X|} \norm{\pi_S - \vecu}^2  \leq \frac{|W|}{|X|} \cdot \frac{\lambda^2 \cdot (\epsilon /\delta)}{|W|} \;=\; \frac{\lambda^2 \cdot (\epsilon /\delta)}{|X|}.
\qedhere
\end{align*}
\end{proof}

In particular, any bi-regular $(\delta, \epsilon)$-extractor is a
$(\delta,\epsilon, \epsilon/\delta)$-fortifier. Hence, if the
underlying graph $G$ of the two-prover game is a $\sqrt{\delta}$-expander,
then \autoref{thm:fortification-on-expanders} states that merely using
an $(\delta,\epsilon)$-extractor as suggested in \cite{Moshkovitz14}
would be sufficient to make it $(\delta,O(\epsilon))$-robust. 

Also, since any graph is trivially a $1$-expander, a bi-regular
$(\delta, \epsilon\delta)$-extractor is also an
$(\delta,\epsilon,\epsilon)$-fortifier. The following lemma also shows
that expanders are also fortifiers with reasonable parameters as well.

\begin{lemma}\label{lem:expanders-are-fortifiers}
  Let $H = ((W,X),E_H)$ be any $\lambda$-expander. Then, for
  every $\delta > 0$, $H$ is also a $(\delta,\sqrt{\lambda^2/\delta}
  ,\lambda^2/\delta)$-fortifier.

  In particular, if $\lambda \leq \epsilon\sqrt{\delta}$, then $H$ is
  an $(\delta, \epsilon,\epsilon)$-fortifier.
\end{lemma}
\begin{proof}
Let $H$ be the normalized adjacency matrix of $H$. Let $S\subseteq W$ such that $|S|\geq \delta|W|$. We have,
$$\|\vecu_S^\perp\|^2\leq \frac{1}{\delta|W|} .$$
Hence, by the expansion property of $H$,
$$\|H\vecu_S -\vecu\|^2 := \|H\vecu_S^\perp\|^2 \leq \lambda^2 \cdot \frac{|W|}{|X|} \cdot \|\vecu_S^\perp\|^2\leq \frac{\lambda^2/\delta}{|X|}.$$
$|H\vecu_S -\vecu|_1 \leq \sqrt{\lambda^2/\delta}$ follows from above and Cauchy-Schwarz inequality.
\end{proof}

Although \autoref{lem:expanders-are-fortifiers} shows that expanders
are also fortifiers for reasonable parameters, the construction in
\autoref{lem:building-fortifiers} is more useful when the underlying
graph for the two-prover game is already a good expander. For example, if the
underlying graph $G$ was a $\delta$-expander, then
\autoref{thm:fortification-on-expanders} suggests that we only require
a
$(\delta,\epsilon,\epsilon/\delta)$-fortifier. \autoref{lem:building-fortifiers}
implies that an $(\delta,\epsilon)$-extractor is already a
$(\delta,\epsilon,\epsilon/\delta)$-fortifier and hence is sufficient
to make the game robust. The main advantage of this is the degree of
$\delta$-expanders must be $\Omega(1/\delta^2)$ whereas we have
explicit $(\delta,\epsilon)$-extractors of degree $(1/\epsilon^2)
\exp(\poly\log\log (1/\delta))$ which has a much better dependence in
$\delta$. This dependence on $\delta$ is crucial for certain
applications. 

\section{Sub-games on large rectangles}\label{sec:large-rectangle-subgames}

Consider a concatenated general game $G^* = H_1 \circ G \circ H_2$ on
$((W,Z),E_{H_1 \circ G \circ H_2})$  and $S\subseteq W$ and $T\subseteq Z$.
Let $\mu_S$ (or $\mu_T$) denote the induced distributions on $X$(or $Y$) obtained by picking
 a uniformly random element of S (or T) and taking a uniformly random
neighbour in $H_1$ (or $H_2$). That is, the degree of any $x\in X$ (or $y\in Y$) within the set $S$ (or $T$) is proportional to $\mu_S(x)$ (or $\mu_T(y)$) (See \autoref{fig:subgames-on-rect}). 

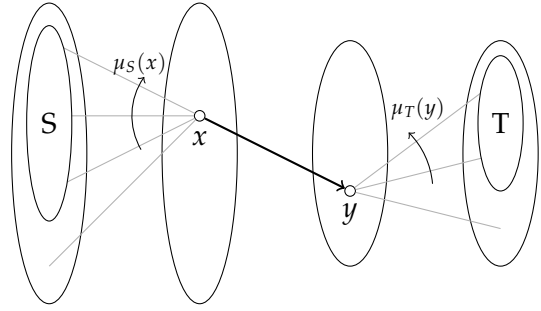
\begin{figure}[h]
\begin{center}
\begin{tikzpicture}
\draw (0,0) ellipse (0.5cm and 2cm);
\draw (2,0) ellipse (0.5cm and 2cm);
\draw (4,0) ellipse (0.5cm and 1.5cm);
\draw (6,0) ellipse (0.5cm and 1.5cm);

\tikzstyle{mycirc}=[circle, draw, inner sep=0pt, minimum width=4pt]
\node[mycirc] (x) at (2,0.5) {};
\node at (2,0.2) {$x$};
\node[mycirc] (y) at (4,-0.5) {};
\node at (4,-0.8) {$y$};

\draw[thick,->] (x) -- (y);

\draw[thin,draw=black!30] (x) -- (0, -1.5);
\draw[thin,draw=black!30] (x) -- (0,-0.5);
\draw[thin,draw=black!30] (x) -- (0, 0.5);
\draw[thin,draw=black!30] (x) -- (0, 1.5);
\draw[<-] (x) ++(145:0.9) arc (145:210:0.9);
\node at (1.2,1.2) {\scriptsize $\mu_S(x)$};

\draw[thin,draw=black!30] (y) -- (6, -1);
\draw[thin,draw=black!30] (y) -- (6, 0);
\draw[thin,draw=black!30] (y) -- (6, 1);
\draw[<-] (y) ++(45:1.1) arc (45:5:1.1);
\node at (4.9,0.6) {\scriptsize $\mu_T(y)$};

\draw[fill=white] (0,0.4) ellipse (0.3cm and 1.3cm);
\node at (0,0.4) {S};

\draw[fill=white] (6,0.4) ellipse (0.3cm and 0.9cm);
\node at (6,0.4) {T};

\end{tikzpicture}
\end{center}
\caption{Sub-games on large rectangles}\label{fig:subgames-on-rect}
\end{figure}
In a subgame $(G^*)_{S\times T}$, the distribution on
verifier checking the underlying constraint on $(x,y)$ is given by the following expression:

\begin{equation}\label{eqn:pi-ST}
\pi_{x,y} \quad = \quad \frac{\mu_S(x)
  \mu_T(y)}{\sum\limits_{(x,y)\in G} \mu_S(x) \mu_T(y)}.
\end{equation}

One way to show that the concatenated game $G^*$ is $(\delta,
O(\epsilon))$-robust would be to show that the above distribution
$\pi_{x,y}$ is $O(\epsilon)$-close to uniform whenever $|S|, |T|$ have
density at least $\delta$ because then the distribution on constraints
that the verifier is going to check in $G^*_{S\times T}$ is
$O(\epsilon)$ close to the distribution on constraints in $G$.  Hence,
up to additive factor of $O(\epsilon)$ the quantity
$\val(G^*)_{S\times T}$ is same as $\val(G)$.  The main question here
what properties should $H_1$ and $H_2$ satisfy so that the above distribution is
close to uniform?

\subsection{Fortifiers are necessary}\label{sec:fortifiers-are-necessary}

To prove that fortifiers are necessary, we shall restrict ourselves to
games on graphs $G = ((X,X),E)$. In such a setting, we can choose to
concatenate with the same graph $H$ both sides. We show that if a
bipartite graph $H = ((W,X),E_H)$, makes a game on a particular graph $G$, $(\delta,
O(\epsilon))$-robust, then $H$ is a good fortifier.

As mentioned earlier, if the graph $G$ had some
expansion properties, then the requirements on the graph $H$ to
concatenate with can be relaxed. Thus, naturally, the worst case graph
$G$ is one that expands the least --- a matching.

\begin{lemma}[Fortifiers are necessary]\label{lem:fortifiers-are-necessary}
  Let $\epsilon, \delta>0$ be small constants.
  Let $H = ((W,X),E_H)$ be a bi-regular graph, and let $G = ((X,X),E)$
  be a matching. Suppose that for every subset $S,T \subseteq W$ with
  $|S|,|T| \geq \delta |W|$, the  distribution (defined in Equation \eqref{eqn:pi-ST}) induced by the game $(H\circ G \circ H)_{S\times T}$
  on the edges of $G$   is $\epsilon$-close to uniform. Then, for every $S \subseteq W$ 
  with $|S| \geq \delta |W|$, 
  \begin{eqnarray}
    \abs{\mu_S - \vecu}_1 & = & \epsilon, \label{eqn:muS-l1}\\
    \norm{\mu_S - \vecu}^2 & = & \frac{O(\epsilon)}{|X|}.\label{eqn:muS-l2}
  \end{eqnarray}
\end{lemma}
\begin{proof}
  It is clear that \eqref{eqn:muS-l1} is necessary as the distribution
  on constraints in the sub-game $(H\circ G\circ H)_{S\times W}$ (as defined
  in \eqref{eqn:pi-ST}) is essentially $\mu_S$ (as $\mu_T$ in this
  case is uniform). 

  As for \eqref{eqn:muS-l2}, let us assume that 
  \[
  \norm{\mu_S - \vecu}^2  \quad=\quad \frac{c}{|X|}.
  \]
  Taking $T=S$, we obtain that  the  distribution (defined in Equation \eqref{eqn:pi-ST}) induced by the game $(H\circ G \circ H)_{S\times S}$
  on the edges of $G$ is given by
  \[
  \pi_{x,x} \quad = \quad \frac{\mu_S(x)^2}{\sum_x \mu_S(x)^2}
  \quad=\quad \pfrac{|X|}{1+c} \cdot \mu_S(x)^2,
  \]
  where the last equality used the fact that $\norm{\mu_S}^2 =
  \norm{\mu_S^\perp}^2 + \norm{\vecu}^2$. 
  \begin{eqnarray*}
    \sum_{x\in X} \abs{\pfrac{|X|}{c+1} \cdot \mu_S(x)^2 \;-\; \frac{1}{|X|}} & = & \pfrac{|X|}{1+c} \cdot \sum_{x\in X} \abs{\mu_S(x)^2 \;-\; \frac{c+1}{|X|^2}}\\
    & = & \pfrac{|X|}{1+c} \cdot \sum_{x\in X} \abs{\mu_S(x) \;-\; \frac{\sqrt{c+1}}{|X|}} \cdot \inparen{\mu_S(x) + \frac{\sqrt{c+1}}{|X|}}\\
    & \geq & \pfrac{1}{\sqrt{1+c}} \cdot \sum_{x\in X} \abs{\mu_S(x) \;-\; \frac{\sqrt{c+1}}{|X|}}\\
    & \geq & \pfrac{1}{\sqrt{1+c}} \cdot \inparen{\inparen{\sqrt{1+c}\;-\;1} \;-\; \sum_{x\in X} \abs{\mu_S(x) \;-\; \frac{1}{|X|}}}\\
    & \geq & \pfrac{1}{\sqrt{1+c}} \cdot \inparen{\inparen{\sqrt{1+c}\;-\;1} \;-\; \epsilon}.
  \end{eqnarray*}
  Thus, if the distribution on constraints is $\epsilon$-close to uniform, then the above lower bound forces $c = O(\epsilon)$. 
\end{proof}

\subsection{General (non-regular) extractors are insufficient}\label{sec:extractor-insufficient}

Suppose $H = ((W,X),E_H)$ is an arbitrary
$(\delta,O(\epsilon))$-extractor. Consider a possible scenario where
there is a subset $S \subseteq W$ with $|S| \geq \delta |W|$ such that
$\mu_S$ is of the form
\[
\mu_S = \inparen{\epsilon,\frac{1-\epsilon}{|X|-1}, \dots, \frac{1 -
    \epsilon}{|X|-1}}.
\]
Notice that this is a legitimate distribution that may be obtained
from a large subset $S$ as $\abs{\mu_S - \vecu}_1$ is easily seen to
be at most $2\epsilon$. However, if $G = ((X,X),E)$ was $d$-regular with $d =
o(|X|)$, then using \eqref{eqn:pi-ST}, the probability mass on the edge $(1,1)$ on the sub-game over
$S\times S$ is
\[
\pi_{1,1} = \pfrac{\epsilon^2}{\epsilon^2 + O\pfrac{\epsilon d}{|X|}}
\approx 1.
\]
In other words, if such a distribution $\mu_S$ can be induced by the
extractor, then the provers can achieve value close to $1$ in the game
$(H\circ G \circ H)_{S\times S}$ by just labelling the edge $(1,1)$
correctly. Thus, $(H\circ G \circ H)$ is not even $(\delta, 0.9)$-robust. 

In \autoref{sec:explicit-extractor-counterexample} we show that we can
adversarially construct a $(\delta,O(\epsilon))$-extractor, although
non-regular, that induces such a skew distribution. In
\autoref{sec:random-graph} we also show that left-regular graphs of
left-degree $o(1/\delta\epsilon)$ are not fortifiers. 

\section{Robustness from fortifiers}

In this section,  we show that concatenating any two-prover game
by fortifier(s) yields a robust game as claimed by
\autoref{thm:fortification-on-expanders}.

\begin{lemma}[Distributions from large rectangles are close to uniform]\label{thm:fortification-on-expanders-restated}
  Let $\mu_S$ and $\mu_T$ be two probability distributions such that 
  \begin{eqnarray}
    \abs{\mu_S^\perp}_1\leq \epsilon_1 &\text{and}& \abs{\mu_T^\perp}_1 \leq \epsilon_1 \label{eqn:mu-perp-L1-bnd},\\
    \norm{\mu_S^\perp}^2 \leq \pfrac{\epsilon_2}{|X|} &\text{and}& \norm{\mu_T^\perp}^2 \leq   \pfrac{\epsilon_2}{|Y|}.\label{eqn:mu-perp-L2-bnd}
  \end{eqnarray}
  Then for any bi-regular graph $G = ((X,Y),E)$ that is a
  $\lambda_0$-expander, the distribution on edge $(x,y)$ (where $x\in
  X$ and $y\in Y$) given by \eqref{eqn:pi-ST} is $(2\epsilon_1 +
  \epsilon_1^2 + 2\lambda_0\cdot \epsilon_2)$-close to uniform.
\end{lemma}

As described in \autoref{sec:large-rectangle-subgames}, if $H_1$ and
$H_2$ are $(\delta,\epsilon_1,\epsilon_2)$-fortifiers, then for any
set $S$ and $T$ of density at least $\delta$, the distribution on the
constraints of $(H_1 \circ G \circ H_2)_{S\times T}$ is given by
\eqref{eqn:pi-ST}. From the above lemma, it follows that the value of
the game on any large rectangle can change only
by the above bound on the statistical distance. By setting the parameters,  \autoref{thm:fortification-on-expanders} follows
immediately from \autoref{thm:fortification-on-expanders-restated}. Further, \autoref{thm:main-thm} also follows from
\autoref{thm:fortification-on-expanders-restated} and \autoref{lem:expanders-are-fortifiers} as any graph is trivially a
$1$-expander.\\

The rest of this section would be devoted to the proof of
\autoref{thm:fortification-on-expanders-restated}. For brevity, let us
assume that $|X| = n$, $|Y| = m$ and let $d$ be the left-degree of $G$. We shall prove
\autoref{thm:fortification-on-expanders-restated} by proving the
following two claims.

\begin{claim}\label{claim:triang1-bound}
\[
\sum_{(x,y)\in G} \abs{\frac{\mu_S(x) \mu_T(y)}{\sum\limits_{(x,y)\in G} \mu_S(x) \mu_T(y)} \; - \; \frac{\mu_S(x) \mu_T(y)}{d / m}} \quad \leq \quad \lambda_0 \cdot \epsilon_2
\]
\end{claim}

\begin{claim}\label{claim:triang2-bound}
\[
\sum_{(x,y)\in G} \abs{\frac{\mu_S(x) \mu_T(y)}{d / m} \; - \; \frac{1}{n\cdot d}} \quad \leq \quad 2\epsilon_1  + \epsilon_1^2 + \lambda_0\cdot  \epsilon_2
\]
\end{claim}

\noindent 
Clearly, \autoref{thm:fortification-on-expanders-restated} follows from
\autoref{claim:triang1-bound} and \autoref{claim:triang2-bound}. 

\begin{proof}[Proof of \autoref{claim:triang1-bound}]
If $G$ denotes the normalized adjacency matrix of the graph $G$ (that is, normalized so that $G\vecu_X = \vecu_Y$), then observe that $\sum_{(x,y)\in G} \mu_S(x) \mu_T(y) = d \cdot \inangle{G\mu_S, \mu_T}$. If we resolve $\mu_S$ and $\mu_T$ in the direction of the uniform distribution and the orthogonal component, we have
\begin{eqnarray*}
\inangle{G\mu_S, \mu_T} & = & \inangle{\vecu_Y, \vecu_Y} \;+\; \inangle{G\mu_S^\perp, \mu_T^\perp} \; = \;\frac{1}{m} \;+\; \inangle{G\mu_S^\perp, \mu_T^\perp} \\
\implies \quad \abs{\inangle{G\mu_S, \mu_T} - \frac{1}{m}} & \leq & \lambda_0 \cdot \norm{\mu_S^\perp} \cdot \norm{\mu_T^\perp}\cdot \sqrt{\frac{n}{m}}\\
& \leq &  \pfrac{\lambda_0 \cdot \epsilon_2}{m}. \quad\quad\text{(using \eqref{eqn:mu-perp-L2-bnd})}
\end{eqnarray*}
Therefore, 
\begin{align*}
\sum_{(x,y)\in G} \abs{\frac{\mu_S(x) \mu_T(y)}{d\inangle{G\mu_S,\mu_T}} \; - \; \frac{\mu_S(x) \mu_T(y)}{d / m}} & \leq  \sum_{(x,y)\in G} \pfrac{\mu_S(x) \mu_T(y)}{d \inangle{G\mu_S,\mu_T}}\abs{1 \;-\;m\inangle{G\mu_S,\mu_T}}\\
& \leq  \lambda_0 \cdot \epsilon_2.\qedhere
\end{align*}
\end{proof}

\begin{proof}[Proof of \autoref{claim:triang2-bound}]
\begin{eqnarray*}
\sum_{(x,y)\in G} \abs{\frac{\mu_S(x) \mu_T(y)}{d / m} \; - \; \frac{1}{n\cdot d}} & = & \pfrac{m}{d}\sum_{(x,y)\in G}\abs{\mu_S(x) \mu_T(y) - \frac{1}{n\cdot m}}.
\end{eqnarray*}
Since $\mu_S(x) = \frac{1}{n} + \mu_S^\perp(x)$ and $\mu_T(y) = \frac{1}{m} + \mu_T^\perp(y)$, 
\begin{eqnarray*}
\pfrac{m}{d}\sum_{(x,y)\in G}\abs{\mu_S(x) \mu_T(y) - \frac{1}{n\cdot m}} & = & \pfrac{m}{d}\sum_{(x,y)\in G}\abs{\frac{\mu_S^\perp(x)}{m} + \frac{\mu_T^\perp(y)}{n} + \mu_S^\perp(x) \mu_T^\perp(y)}\\
\text{(Using triangle inequality)}& \leq & \frac{1}{d}\sum_{(x,y)\in G}\abs{\mu_S^\perp(x)} + \frac{m}{nd}\sum_{(x,y)\in G}\abs{\mu_T^\perp(y)} \\
& & \quad\;\quad\;\quad \;+\; \pfrac{m}{d}\sum_{(x,y)\in G} \abs{\mu_S^\perp(x) \mu_T^\perp(y)}\\
& = & \abs{\mu_S^\perp}_1 \;+ \; \abs{\mu_T^\perp}_1 + \pfrac{m}{d}\sum_{(x,y)\in G} \abs{\mu_S^\perp(x) \mu_T^\perp(y)},
\end{eqnarray*}
where the last equality uses the fact that $G$ is a bi-regular graph. Define $f_S (x) \equiv |\mu_S^\perp(x)|$ is a vector with the entrywise absolute values of $\mu_S^\perp$, and similarly $f_T$. Then, the RHS above equation reduces to 
\begin{eqnarray}
\abs{\mu_S^\perp}_1 \;+ \; \abs{\mu_T^\perp}_1 + \pfrac{m}{d}\sum_{(x,y)\in G} \abs{\mu_S^\perp(x) \mu_T^\perp(y)}& = & \abs{\mu_S^\perp}_1 \;+ \; \abs{\mu_T^\perp}_1 \nonumber \\ 
& & \quad \;+\; \pfrac{m}{d}\cdot \sum_{(x,y)\in G} f_S(x) f_T(y) \nonumber \\
& = & \abs{\mu_S^\perp}_1 \;+ \; \abs{\mu_T^\perp}_1 \;+\; m \inangle{Gf_S, f_T}\nonumber \\
\text{(Using \eqref{eqn:mu-perp-L1-bnd})}\quad\quad\quad& \leq & 2\epsilon_1  \;+\; m \cdot \inangle{Gf_S, f_T}.\label{eqn:better-fsGft-bnd}
\end{eqnarray}
A simple bound for $m\cdot \inangle{Gf_S, f_T}$ would $m \norm{G\mu_S^\perp}\norm{\mu_T^\perp}$ by Cauchy-Schwarz inequality. We can use the expansion of $G$ again to estimate this better. Consider the decomposition $f_S = \alpha_1 \cdot \vecu_X + f_S^\perp$ and $f_T = \alpha_2 \cdot \vecu_Y + f_T^\perp$. It follows that $\alpha_1 = \abs{f_S}_1$ and $\alpha_2 = \abs{f_T}_1$, and hence $\alpha_1, \alpha_2 \leq \epsilon_1$ by \eqref{eqn:mu-perp-L1-bnd}. Hence,
\begin{eqnarray*}
m \cdot \inangle{Gf_S,f_T} & = & \alpha_1 \cdot \alpha_2  \;+\; m\cdot \inangle{Gf_S^\perp, f_T^\perp}\\
& \leq &  \epsilon_1^2 \;+\; m\cdot \lambda_0 \cdot \norm{f_S^\perp} \cdot \norm{f_T^\perp} \cdot \sqrt{\frac{n}{m}}\\
& \leq &  \epsilon_1^2 \;+\; m\cdot \lambda_0 \cdot \norm{\mu_S^\perp} \cdot \norm{\mu_T^\perp}\cdot \sqrt{\frac{n}{m}}\\
\text{(Using \eqref{eqn:mu-perp-L2-bnd})}\quad& \leq & \epsilon_1^2 \;+\;\lambda_0 \epsilon_2.
\end{eqnarray*}
Combining this with \eqref{eqn:better-fsGft-bnd}, we get
\[
\sum_{(x,y)\in G} \abs{\frac{\mu_S(x) \mu_T(y)}{d / m} \; - \; \frac{1}{n\cdot d}} \quad\leq \quad 2\epsilon_1  + \epsilon_1^2 + \lambda_0 \epsilon_2.\qedhere
\]
\end{proof}

\subsection*{Acknowledgements}

We would like to thank Dana Moshkovitz for several discussions and
clarifications regarding the initial counter-example. We would also
like to thank Mohammad Bavarian for pointing out that our proof might
generalize for general two-prover games, and would like to thank Anup
Rao for pointing out subtleties involving parallel repetition for
general games. We also would like to thank Prahladh Harsha, Irit Dinur
and Amir Shpilka for many fruitful conversations and comments on the
write-up.

\bibliographystyle{prahladhurl}
\bibliography{references}

\appendix

\section{An explicit extractor that does not provide robustness} \label{sec:explicit-extractor-counterexample}
Let $H = ((W,X),E_H)$ be any $(\delta,\epsilon)$-extractor. Let us
assume that the extractor is left-regular with left-degree $D$, and
let $m = |W|$ and $n = |X|$. For any $x\in X$ and $S\subseteq W$, let
$d_S(x)$ denote the degree of $x$ in $S$. Let us fix one $S \subset W$
such that $|S| = \delta |W|$.

We will transform the graph $H$ so that the distribution induced by
the set $S$ looks like the counter-example described in
\autoref{sec:extractor-insufficient} in the following two steps by
altering the edges in the subgraph $S \times X$:
\begin{enumerate}
\item First change the degree into $X$ from $S$ to be exactly uniform. 
\item Next further change the degrees into $X$ from $S$ to be like the 
counterexample
\end{enumerate}
Both these operations can be achieved in a monotone fashion: for every
$x\in X$, the neighborhood of every vertex is either a superset, or a
subset of its neighborhood before each
operation. 

We will show that moving the edges this way does not perturb the
indegree distribution from other large sets by too much, and the
resulting graph is a $(\delta, O(\epsilon))$ extractor as long as the
number of edges we relocate is at most $O(\epsilon\delta\cdot
mD)$. This process will preserve the left-regularity of $H$ but would
\emph{not} preserve bi-regularity.\\

First let us move edges (monotonically) from $S$ into $X$ create the
uniform distribution on $X$. When doing this, the degree of each
vertex changes by $\Delta_S(x) := |d_S(x)-\frac{\delta m D}{n}|$, where
$d_S(x)$ was the old degree. From the extractor property, we know
that:
\begin{equation}
\label{eqn:extractor-eq1}
\sum_{x\in X} \Delta_S (x) \;=\; \sum_{x\in X} (\delta m D )\abs{\frac{d_S(x)}{\sum d_S(x)} - 
\inparen{\frac{1}{n}}}\spaced{\leq} \epsilon \delta \cdot m D.
\end{equation}

Every vertex $x\in X$ now has degree $d^S_\mathrm{avg}$. Fix some
vertex $x_1 \in X$, and relocate from every other $x \neq x_1$ any set
of $\epsilon \cdot d^S_\mathrm{avg}$ edges to be incident on
$x_1$. Thus, if $d_S'(x)$ refers to the new degrees, we have
$d_S'(x_1)$ is $(1 + \epsilon n) d^S_{\mathrm{avg}}$ where as
$d_S'(x)$ is $(1-\epsilon) d^S_\mathrm{avg}$ for every other $x\neq
x_1$.

The further change in degrees incurred on any $x\in
X$ is $\Delta'_S(x) := \abs{d_S'(x) - \frac{\delta m D}{n}}$. Since we
this process only relocates $O(\epsilon \cdot  d^S_{\mathrm{avg}} |X|)$ edges, we have
\begin{equation}
\label{eqn:extractor-eq2}
\sum_{x\in X} \Delta'_S(x) \spaced{=} \sum_{x\in X} \abs{d_S'(x) - d^S_{\mathrm{avg}}}  \spaced{\leq} 
O(n \cdot \epsilon \cdot d^S_{\mathrm{avg}}) \spaced{=}O(\epsilon\delta \cdot mD).
\end{equation}

Thus, the neighbourhood of any vertex $x$ has changed additively by at most
$\Delta_S(x)+\Delta'_S(x)$. Therefore, for any subset $T \subseteq
W$ of size at least $\delta |W|$,
\begin{eqnarray*}
\sum_{x\in X} \abs{d'_T(x) - d^T_{\mathrm{avg}}} &\leq& \sum_{x\in X} \abs{d_T(x) - 
d^T_{\mathrm{avg}}} \spaced{+} \sum_{x\in X} \abs{d'_T(x) - d_T(x) }\\
& \leq & \epsilon |T| D \spaced{+} \sum_{x\in X} \inparen{\Delta_S(x)+\Delta'_S(x)}\\
& \leq & \epsilon |T| D \spaced{+} O(\epsilon\delta \cdot m D) \quad 
(\text{using \eqref{eqn:extractor-eq1} and \eqref{eqn:extractor-eq2}}) \\
&\leq & O(\epsilon \cdot |T| D).
\end{eqnarray*}
Thus, the new graph after relocating edges is still an
$(\delta,O(\epsilon))$-extractor. This extractor, induces a
distribution similar to the one described in
\autoref{sec:extractor-insufficient} and hence cannot provide
robustness.

\section{Lower bounds on degree of fortifiers}\label{sec:random-graph}

In this section, we will show that an attempt to make a game $(\delta,
\epsilon)$-robust by concatenating any left-regular graph with left degree $D$
fails if $D \leq o(1/\epsilon\delta)$. 

\begin{lemma}\label{lem:random-graph-main-lemma}
  Let $H = ((W,X),E_H)$ be a left-regular bipartite graph with
  left-degree $D = 1/(c\cdot \epsilon\delta)$ for some $c > 0$, and small enough
  constants $\epsilon,\delta$. Then, there exists a subset
  $S\subseteq W$ with $|S| \geq \delta |W|$ such that if $p$ was the
  distribution on $X$ induced by the uniform distribution on $S$ then 
  \[
  \norm{p - \vecu}^2 \;\geq\; \frac{\Omega(c\epsilon)}{|X|}.
  \]
\end{lemma}
\begin{proof}
  Let $d_\mathrm{avg} = |W|D/|X|$. Note that at most $|X|/2$ vertices
  $x$ satisfy $\deg(x) \geq 2d_\mathrm{avg}$. Further, if there is a
  set $S$ of $|X|/4$ vertices $x$ that $\deg(x) < (0.5)
  d_\mathrm{avg}$, then if $p$ is the distribution on $X$ induced by
  the uniform distribution on $W$, then $|p -\vecu|_1 > 1/4$ which implies that $\norm{p - \vecu}_2^2 \geq \frac{1}{4|X|}$ by Cauchy-Schwarz. \\

  Otherwise, there exists $X' \subset X$ such that $|X'| =c\,\epsilon
  \delta^2 |X|$ and for each $x\in X'$ we have $(0.5)d_\mathrm{avg} < \deg(x) < 2 d_\mathrm{avg}$. Consider the set $S_0$
  of all neighbours of $X'$. If $D < 1/(c\epsilon\delta)$, we have $|S_0| \leq 2c\,\delta^2 \epsilon\cdot|W|D = 2\delta |W|$ which is a very small fraction of $|W|$
  when $\delta$ is small enough. Consider an arbitrary set
  $S_1\subseteq W$ such that $|S_1| = \delta m$, with $S_1\cap S_0 =
  \emptyset$. Let $S_2 = S_0 \cup S_1$. Let $\pi_1, \pi_2$ be the
  probability distribution on $X$ induced by $S_1,S_2$
  respectively. Note that $|S_2| \leq 3\delta |W|$.

  For every $x\in X'$, we know that $\pi_1(x) = 0$ and $\pi_2(x) =
  \Omega\pfrac{1}{\delta |X|}$. Therefore,
  \[
  \norm{\pi_1 - \pi_2}^2 \geq \Omega\pfrac{c \delta^2 \epsilon
    |X|}{\delta^2 |X|^2} \;=\; \frac{\Omega(c \epsilon)}{|X|}.
  \]
  Since $\norm{\pi_1 - \pi_2} \leq \norm{\pi_1 - \vecu} + \norm{\pi_2
    - \vecu}$, we have that one of the sets $S_1$ or $S_2$ shows the
  validity of the lemma
\end{proof}

We thus immediately infer the following:
\begin{corollary} \label{cor:bireg-not-fortifier}
  For all small enough $\delta,\epsilon > 0$, no left-regular graph $H =
  ((W,X),E_H)$ with left-degree $D = o(1/\epsilon\delta)$ is an
  $(\delta, * ,\epsilon)$-fortifier.
\end{corollary}

Note that any $(\delta,\epsilon,\epsilon)$-fortifier is in particular an $(\delta,\epsilon)$-extractor, and hence we also have that $D = \Omega((1/\epsilon^2) \log(1/\delta))$ \cite{RT00}. We also point out that the construction of
\autoref{lem:expanders-are-fortifiers} has left-degree $D =
\tilde{O}(1/\epsilon^2\delta).$ The above essentially shows this
construction is almost optimal.

\section{Parallel repetition from fortification}

We present a mild generalization of \autoref{thm:robustness-to-pr} to
general bi-regular games, following essentially the same strategy as
in \cite{Moshkovitz14}.

\begin{lemma}
  \label{lem:p2-generalgames}
  Let $G = ((X,Y),E)$ be a $(\delta,\epsilon)$-robust general game that is bi-regular with $2\delta \inparen{|\Sigma_X||\Sigma_Y|}^{k-1} < \epsilon$. Then, 
  \[
  \val(G^k) \quad \leq \quad \val(G^{k-1}) \cdot \inparen{\val(G) + \epsilon} \;+\; \epsilon.
  \]
\end{lemma}
\begin{proof}
Consider any deterministic strategy for the provers. These are merely functions 
\begin{align*}
f_1:X^k & \rightarrow \Sigma_X^k \quad \text{and} \quad f_2:Y^k  \rightarrow \Sigma_Y^k
\end{align*}
that assign labels to the $k$ queries asked by the verifier. For every
$(k-1)$-tuple of queries $\bar{v} = (v_1,\dots, v_{k-1})$ with each
$v_i := (x_i ,y_i) \in E$, and an arbitrary tuple of $(k-1)$ pairs of labels
$\bar{\sigma}:=((\sigma_1,\sigma_1'),\ldots,
(\sigma_{k-1},\sigma_{k-1}')) \in \inparen{\Sigma_X \times \Sigma_Y}^{k-1}$,
define the rectangle $\mathcal{R}_{\bar{v},\bar{\sigma}} :=
S_{\bar{v},\bar{\sigma}} \times T_{\bar{v},\bar{\sigma}}$ where
\begin{align*}
S_{\bar{v},\bar{\sigma}} & = \setdef{x_k}{f_1(x_1,x_2,\dots, x_k)    \  \text{assigns label $\sigma_i$ to $x_i$ for all $i\leq k-1$}},\\
T_{\bar{v},\bar{\sigma}} & = \setdef{y_k}{f_2(y_1,y_2,\dots, y_k)\  \text{assigns label $\sigma_i'$ to $y_i$ for all $i\leq k-1$}}.
\end{align*}
Also we shall call a rectangle $\mathcal{R}_{\bar{v},\bar{\sigma}}$ \emph{accepting} if every coordinate $(\sigma_i, \sigma_i')$ of $\bar{\sigma}$ satisfies the constraint on $v_i = (x_i,y_i)$ for all $1\leq i \leq k-1$. In words, an accepting rectangle $\mathcal{R}_{\bar{v},\bar{\sigma}}$ is the set of all possible queries $v_k$ for the last round such that the provers win on the first $(k-1)$
rounds with $x_1,\dots,x_{k-1}$ and $y_1,\dots,y_{k-1}$ getting labels $\sigma_1,\dots,\sigma_{k-1}$ and $\sigma_1',\dots, \sigma_{k-1}'$ respectively.  We shall call a rectangle $\mathcal{R}_{\bar{v},\bar{\sigma}}$ ``large'' if $S_{\bar{v},\bar{\sigma}}$ and $T_{\bar{v},\bar{\sigma}}$ have density at least $\delta$, and ``small'' otherwise.
We shall partition the space of all possible queries $(v_1,\dots, v_k)$ into the following sets.  Note that $v_k$ belongs to a unique rectangle $ \mathcal{R}_{\bar{v},\bar{\sigma}}$.
\begin{itemize}
\item $\mathcal{A}_0 = \setdef{(v_1,\dots, v_k)}{ \mathcal{R}_{\bar{v},\bar{\sigma}}\text{ is not accepting}}$
\item $\mathcal{A}_1 = \setdef{(v_1,\dots, v_k)}{\mathcal{R}_{\bar{v},\bar{\sigma}} \text{ is accepting and ``large''}}$
\item $\mathcal{A}_2 = \setdef{(v_1,\dots, v_k)}{\mathcal{R}_{\bar{v},\bar{\sigma}} \text{ is accepting and ``small''}}$
\end{itemize}

Observe that $\abs{\mathcal{A}_1} + \abs{\mathcal{A}_2} \leq
\val(G^{k-1}) \cdot \abs{E}^k$ because $\mathcal{A}_1 \union
\mathcal{A}_2$ is the set of queries on which the provers succeed on
the first $(k-1)$ rounds. 

Also, the projection of elements in set $\mathcal{A}_1$ to the $k$th coordinate,
 is essentially a union of large rectangles. By the $(\delta, \epsilon)$-robustness of $G$,
any strategy of the provers can succeed on each large rectangle with
probability at most $\val(G) + \epsilon$. Hence, the provers succeed
on at most a $(\val(G) + \epsilon)$-fraction of points in
$\mathcal{A}_1$.

Furthermore, since $G$ is regular, we get
$\abs{\mathcal{A}_2}$ is at most $\abs{E}^{k-1} \cdot 2\delta|E|\cdot
\abs{\Sigma_X\times \Sigma_Y}^{k-1} \leq \epsilon \abs{E}^k$ by the
choice of $\delta$ and $\epsilon$.\footnote{In the case of projection
  games, the set of $\bar{\sigma}$ that are accepting pairs for
  $\bar{v}$ can be indexed with $\Sigma_Y^{k-1}$ instead of $(\Sigma_X
  \times \Sigma_Y)^{k-1}$, and that gets the better parameters for
  projection games as in \autoref{thm:robustness-to-pr}.}

Hence, the total number of queries on which the provers can succeed is
upper bounded by $(\val(G) + \epsilon) \abs{\mathcal{A}_1} +
\abs{\mathcal{A}_2}$. It therefore follows that they succeed on at
most a $\val(G^{k-1}) (\val(G) + \epsilon) + \epsilon$ fraction of
queries.
\end{proof}

Unfolding the recursion from the above lemma, we get the following generalization of \autoref{thm:robustness-to-pr}.

\begin{corollary}
Let $G = ((X,Y),E)$ be a $(\delta,\epsilon)$-robust general game with $2\delta \inparen{|\Sigma_X||\Sigma_Y|}^{k-1} < \epsilon$. Then, 
\[
\val(G^k) \leq (\val(G) + \epsilon)^k + k\cdot \epsilon.
\]
\end{corollary}

As mentioned earlier, this is not useful when say $\abs{\Sigma_X} =
\exp(O(1/\delta))$, which is unfortunately the case when an arbitrary
game is made robust by concatenating with a fortifier. 

\section{Making the graph bi-regular}\label{sec:deg-reduction}

In this section, we shall show that a general game on a graph can be converted to a slightly larger game on a bi-regular
graph with almost the same value.

\begin{lemma}\label{lem:general-game-biregular}
Given a two-prover game  $G$  any graph $((X,Y),E)$. For every $\epsilon > 0$, there is a polynomial time algorithm to construct a game $G'$ with $\mathrm{size}(G') = \mathrm{size}(G) \cdot \tilde{O}\inparen{(\abs{\Sigma_X} + \abs{\Sigma_Y})/\epsilon}^5$ such that  $G'$ is on a bi-regular graph and  $\val(G') \leq \val(G) + \epsilon$. 
\end{lemma}

The rest of this section would be a proof of this. Suppose we have a
graph $G = ((X,Y),E)$ that is possibly non-regular. We shall make some
transformations on the graph to make it bi-regular such that it does
not affect the value of the game by much. This is along the same lines
as the technique used by Dinur and Harsha~\cite{DH13}. We shall need
the following well-known \emph{Expander Mixing Lemma}.

\begin{lemma}[Expander Mixing Lemma] \label{lem:expander-mixing} Let
  $H = ((P,Q),E)$ be a $\lambda$-expander with $|P| = |Q|$. Then, for
  every subsets $A \subseteq P$ and $B\subseteq Q$,
  \[
  \abs{\frac{\abs{E(A,B)}}{\abs{E}} \;-\; \frac{|A|}{|P|} \cdot \frac{|B|}{|Q|} } \;\leq\; \lambda
  \]
\end{lemma}
A proof of the above lemma may be found in any text that studies
expanders graphs (for example, \cite[Chapter 5]{AS92}).

We shall make the graph bi-regular in two steps. We shall first make a
transformation that makes it regular on the right side, and then
repeat the same process on the left. But first, we would need to
ensure that the degree on the $Y$ side is large enough for the
transformation to work. This is just done by creating $d$ copies of
every edge with the same constraint. The graph therefore becomes a
multi-graph but the value remains the same.\footnote{One could also do
  this by replicating every vertex $d$ times and adding the edges
  between them.}

Thus, from now on, we assume that we are given a game $G = ((X,Y),E)$,
with the minimum degree being ``large enough'', that we want to make
biregular. The transformation of $G$ to make it regular on right side
is as follows (\autoref{fig:biregularity}):

For every vertex $y\in X$ with degree $d_y$, we shall have a set $C_y$
of $d_y$ vertices. Between the vertices $C_y$ and the neighbourhood of
$y$ (in $G$), we shall add a $\lambda$-expander of degree $d$. The constraint on any
edge between $x \in N(y)$ and a vertex in $C_y$ would be the same as
$\psi_{(x,y)}$. Let us denote this game by $G_\lambda$.
\begin{figure}[h]
\begin{center}
\begin{tikzpicture}
\tikzstyle{mycirc}=[circle, draw, inner sep=0pt, minimum width=4pt]
\draw (0,0) ellipse (0.5cm and 1.5cm);
\node at (0,-1.8) {\small $N(y)$};
\node[mycirc] (x1) at (0,1) {};
\node[mycirc] (x2) at (0,0.5) {};
\node[mycirc] (x3) at (0,0) {};
\node[mycirc] (x4) at (0,-0.5) {};
\node[mycirc] (x5) at (0,-1) {};

\node[mycirc] (y) at (2,0) {}
edge[draw=red!80] (x1)
edge[draw=blue!80] (x2)
edge[draw=green!80] (x3)
edge[draw=brown] (x4)
edge[draw=purple] (x5);
\node at (2,-0.3) {\small $y$};

\draw (5,0) ellipse (0.5cm and 1.5cm);
\node at (5,-1.8) {\small $N(y)$};

\node[mycirc] (y1) at (7,1) {};
\node[mycirc] (y2) at (7,0.5) {};
\node[mycirc] (y3) at (7,0) {};
\node[mycirc] (y4) at (7,-0.5) {};
\node[mycirc] (y5) at (7,-1) {};

\node[mycirc] (x11) at (5,1) {}
edge[draw=red!80] (y2)
edge[draw=red!80] (y4)
;
\node[mycirc] (x21) at (5,0.5) {}
edge[draw=blue!80] (y1)
edge[draw=blue!80] (y5)
;
\node[mycirc] (x31) at (5,0) {}
edge[draw=green!80] (y3)
edge[draw=green!80] (y2)
;
\node[mycirc] (x41) at (5,-0.5) {}
edge[draw=brown] (y1)
edge[draw=brown] (y5);
\node[mycirc] (x51) at (5,-1) {}
edge[draw=purple] (y3)
edge[draw=purple] (y4);
\draw (7,0) ellipse (0.5cm and 1.5cm);
\node at (7,-1.8) {\small $C_y$};
\draw (3,0) edge[ultra thick,->] (4,0);
\end{tikzpicture}
\end{center}
\caption{Enforcing bi-regularity}\label{fig:biregularity}
\end{figure}
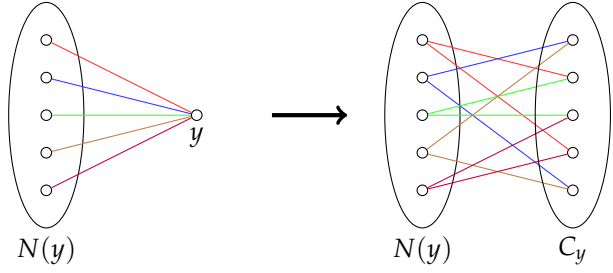
\begin{lemma}
$\val(G_\lambda) \leq \val(G) + \lambda \abs{\Sigma_Y}$. 
\end{lemma}
\begin{proof}
  Consider any labelling $L_\lambda$ of $G_\lambda$. From this, let
  $L$ be the natural randomized labelling for $G$ such that $L(x) =
  L_\lambda(x)$ for every $x\in X$, and $L(y) = L_\lambda(y_i)$ be
  where $y_i$ is a random element of $C_y$. For every $y\in Y$, let
  $\delta_y$ be the expected fraction of edges incident on $y$ that are
  satisfied by this assignment.
  \[
  \delta_y = \sum_{\sigma \in \Sigma_Y} \Pr[L(y) = \sigma] \cdot \Pr_{x \sim y}[(L(x),\sigma)\text{ satisfies $\psi_{(x,y)}$} ]
  \]
  By the definition of $\val(G)$, we know that $\sum\limits_{y\in Y} d_y\delta_y \leq \val(G)\cdot |E|$. 
  \begin{addmargin}[2em]{5em}
    \begin{subclaim}\label{sclm:sampler-mixing}
      For every $y\in Y$, the fraction of edges between $N(y)$ and $C_y$ that are satisfied by $L_\lambda$ is at most $(\delta_y + \lambda |\Sigma_y|)$
    \end{subclaim}
  \end{addmargin}
  
  \medskip
  
  \noindent 
  Before we prove this, let us see why this is sufficient to complete
  the proof of the lemma. The number of edges between $C_y$ and $N(y)$
  is exactly $d \cdot d_y$ where $d$ is the degree of the
  expander. Therefore, the number of edges in $G_\lambda$
  that are satisfied is
  \begin{eqnarray*}
    \sum_{y \in Y} d \cdot d_y \cdot (\delta_y + \lambda|\Sigma_y|) & \leq & d \cdot \sum_{y\in Y} d_y \delta_y \;+\; O(d \lambda \abs{\Sigma_Y}) \cdot \sum_{y \in Y} d_y\\
    & \leq & \inparen{\val(G) + \lambda \abs{\Sigma_Y}} \cdot \abs{E_\lambda}
  \end{eqnarray*}
  as claimed by the lemma. Thus, it suffices to prove \autoref{sclm:sampler-mixing}. 

  \begin{addmargin}[2em]{5em}
    \begin{myproof}{\autoref{sclm:sampler-mixing}}
      The number of edges between $C_y$ and $N(y)$ is $d \cdot
      d_y$. Partition the vertices of $C_y$ into sets
      $\setdef{C_{y,\sigma}}{\sigma \in \Sigma_Y}$ based on the label
      assigned by $L_\lambda$. For every $\sigma\in \Sigma_Y$, let
      $A_\sigma$ denote the set of vertices $x\in N(y)$ such that
      $(L_\lambda(x),\sigma)$ satisfies $\psi_{(x,y)}$. Hence, the set
      of edges that are satisfied by $L_\lambda$ is precisely $\Union_{\sigma} E(A_\sigma, C_{y,\sigma})$. By \autoref{lem:expander-mixing}, 
      \begin{eqnarray*}
        \abs{E(A_\sigma,C_{y,\sigma})} & \leq & \abs{A_\sigma} \cdot \abs{C_{y,\sigma}} \cdot \frac{d}{d_y} \; + \; \lambda \cdot d \cdot d_y\\
        \implies \sum_{\sigma \in \Sigma_Y}\abs{E(A_\sigma,C_{y,\sigma})} & \leq & \sum_{\sigma \in \Sigma_Y} \abs{A_\sigma} \cdot \abs{C_{y,\sigma}} \cdot \frac{d}{d_y} \;+\; \lambda \cdot \abs{\Sigma_Y} \cdot d \cdot d_y\\
        & = & (d \cdot d_y) \sum_{\sigma \in \Sigma_Y} \Pr[L(y) = \sigma] \cdot \Pr_{x \sim y}[(L(x),\sigma)\text{ satisfies $\psi_{(x,y)}$} ]\\
        &  & \quad + \quad \lambda \abs{\Sigma_Y} \cdot d \cdot d_y\\
        & = & \inparen{\delta_y \;+\; \lambda \cdot \abs{\Sigma_Y}}  \cdot (d\cdot d_y)
      \end{eqnarray*}
      as claimed, since the number of edges is $d \cdot d_y$. 
    \end{myproof}
  \end{addmargin}
  
  \noindent 
  That hence finishes the proof of the Lemma. 
\end{proof}

This operation ensures that the right-degree of the game $G_\lambda$
is $d$ and the value changes by at most $\epsilon/2$ if $\lambda <
(\epsilon/2 \abs{\Sigma_Y})$. By \autoref{lem:expander-construction},
we can choose explicit constructions of expanders with $d =
\tilde{O}(1/\lambda^2) = \tilde{O}((\abs{\Sigma_Y}/\epsilon)^2)$. The
graph is now right-regular with degree $d$, and the degree of every
$x\in X$ has increased by a factor of $d$. Repeating the same process
for the other side makes both sides regular and the value changes by
at most $\epsilon$. \qed (\autoref{lem:general-game-biregular})

\end{document}